\documentclass{article}
\usepackage[utf8]{inputenc}

\usepackage[text={15.4cm,22cm},centering]{geometry}

\usepackage{amsfonts}
\usepackage{latexsym}
\usepackage{amsmath}
\usepackage{amssymb}
\usepackage[mathcal]{euscript}
\usepackage{graphicx,color,subfigure}
\usepackage{hyperref}

\newtheorem{fact}{Fact}
\newtheorem{theorem}[fact]{Theorem}
\newtheorem{corollary}[fact]{Corollary}
\newtheorem{definition}[fact]{Definition}
\newtheorem{proposition}[fact]{Proposition}
\newtheorem{lemma}[fact]{Lemma}
\newtheorem{example}[fact]{Example}
\newtheorem{remark}[fact]{Remark}

\newenvironment{proof}[1][Proof]{\emph{#1.} }{\hfill{$\Box$}\\}

\title{Quantifiers closed under partial polymorphisms} 

\author{Anuj Dawar\\Department of Computer Science and Technology\\ University of Cambridge, UK\\anuj.dawar@cl.cam.ac.uk
  \and Lauri Hella \\Faculty of Information Technology and
  Communication Sciences\\  Tampere University,
  Finland\\lauri.hella@tuni.fi}

\DeclareMathOperator{\N}{\ensuremath{\mathbb{N}}}
\DeclareMathOperator{\Z}{\ensuremath{\mathbb{Z}}}

\newcommand{\dom}{\mathrm{dom}}
\newcommand{\ran}{\mathrm{rng}}
\newcommand{\ar}{\mathrm{ar}}

\newcommand{\swn}{\mathrm{swn}}
\newcommand{\Odd}{\mathrm{Odd}}

\newcommand{\cI}{\mathcal{I}}

\newcommand{\cK}{\mathcal{K}}

\newcommand{\cM}{\mathcal{M}}
\newcommand{\cMaj}{\mathcal{MJ}}

\newcommand{\cN}{\mathcal{N}}
\newcommand{\cP}{\mathcal{P}}

\newcommand{\gb}{\mathrm{GB}}

\newcommand{\odd}{\mathrm{od}}
\newcommand{\even}{\mathrm{ev}}

\newcommand{\tw}{\mathrm{tw}}

\newcommand{\FO}{\mathrm{FO}}

\newcommand{\bQ}{\mathbf{Q}}

\newcommand{\CSP}{\mathsf{CSP}}
\newcommand{\bCSP}{\mathbf{CSP}}
\newcommand{\csp}{\mathrm{CSP}}

\newcommand{\PG}{\mathrm{PG}}

\newcommand{\CR}{\mathrm{CR}}

\newcommand{\PI}{\mathrm{PI}}
\newcommand{\mA}{\mathfrak{A}}

\newcommand{\mB}{\mathfrak{B}}

\newcommand{\mC}{\mathfrak{C}}
\newcommand{\mD}{\mathfrak{D}}
\newcommand{\mE}{\mathfrak{E}}
\newcommand{\mH}{\mathfrak{H}}

\newcommand{\PTIME}{\mathrm{PTIME}}

\newcommand{\Zp}[1]{\ensuremath{\Z /  #1\Z}}

\newcommand{\ra}{\rightarrow}

\begin{document}

\maketitle

\begin{abstract}
We study Lindström quantifiers that satisfy certain closure properties 
which are motivated by the study of polymorphisms in the context of 
constraint satisfaction problems (CSP).  When the algebra of 
polymorphisms of a finite structure $\mB$ satisfies certain equations, this 
gives rise to a natural closure condition on the class of structures 
that map homomorphically to $\mB$.  The collection of quantifiers that 
satisfy closure conditions arising from a fixed set of equations are 
rather more general than those arising as CSP.  For any such conditions
$\cP$, we define a pebble game that delimits the distinguishing power of the 
infinitary logic with all quantifiers that are $\cP$-closed.  We use the 
pebble game to show that the problem of deciding whether a system of
linear equations is solvable in $\Zp{2}$ is not expressible in the infinitary logic with all 
quantifiers closed under a near-unanimity condition.
\end{abstract}

\section{Introduction}\label{sec:Intro}

Generalized quantifiers, also known as Lindstr\"{o}m quantifiers, have played a significant role in the development of finite model theory.  The subject of finite model theory is the expressive power of logics in the finite, and Lindstr\"om quantifiers provide a very general and abstract method of constructing logics.  We can associate with any isomorphism-closed class of structures $\cK$, a quantifier $Q_{\cK}$ so that the extension $L(Q_{\cK})$ of a logic $L$ with the quantifier $Q_\cK$ is the \emph{minimal} extension of $L$ that can express the class $\cK$, subject to certain natural closure conditions.  For this reason, comparing the expressive power of logics with Lindstr\"om quantifiers is closely related to comparing the descriptive complexity of the underlying classes of structures.

Another reason for the significance of Lindstr\"om quantifiers is that we have powerful methods for proving inexpressibility in logics with such quantifiers.  In particular, games, based on Hella's bijection games~\cite{Hella96}, are the basis of the most common inexpressivity results that have been obtained in finite model theory. The $k,n$-bijection game was introduced by Hella to
characterize equivalence in the logic $L^{k}_{\infty\omega}(\bQ_n)$,
which is the extension of the infinitary logic with $k$ variables by means of all $n$-ary Lindstr\"{o}m quantifiers.  A quantifier $Q_\cK$ is $n$-ary if the class $\cK$ is defined over a vocabulary $\sigma$ in which all relation symbols have arity $n$ or less.    In
particular, the $k,1$-bijection game, often called the $k$-pebble
bijection game, characterizes equivalence in
$L^{k}_{\infty\omega}(\bQ_1)$ which has the same expressive power as
$C^{k}_{\infty\omega}$, the $k$-variable infinitary logic with
counting.  Hella uses the $k,n$-bijection game to show that, for each
$n$, there is an $(n+1)$-ary quantifier that is not definable in
$L^{k}_{\infty\omega}(\bQ_n)$ for any $k$.

The $k,1$-bijection game has been widely used to establish inexpressibility results for $C^{k}_{\infty\omega}$.  The
$k,n$-bijection game for $n> 1$ has received relatively less
attention.  One reason is that, while equivalence in
$C^{k}_{\infty\omega}$ is a polynomial-time decidable relation, which
is in fact a relation much studied on graphs in the form of the
Weisfeiler-Leman algorithm, in contrast the relation induced by the
$k,n$-bijection game for $n> 1$ reduces to isomorphism on graphs and
is intractable in general.  Nonetheless, there is some interest in
studying, for example, the non-trivial equivalence induced by
$L^{k}_{\infty\omega}(\bQ_2)$ on structures with a ternary relation.
Grochow and Levet~\cite{GL22} investigate this relation on finite
groups.

A second reason why the logics $L^{\omega}_{\infty\omega}(\bQ_n)$ have attracted less
interest is that in finite model theory we are often interested in
logics that are closed under vectorized first-order interpretations.  This is
especially so in descriptive complexity as the complexity classes we
are trying to characterize usually have these closure properties.
While $L^{\omega}_{\infty\omega}(\bQ_1)$ is closed under first-order
interpretations, this is not the case for
$L^{\omega}_{\infty\omega}(\bQ_n)$ for $n> 1$.  Indeed, the closure of
$L^{\omega}_{\infty\omega}(\bQ_2)$ under interpretations already includes
$\bQ_n$ for all $n$ and so can express all properties of finite
structures.  So, it seems that beyond
$L^{\omega}_{\infty\omega}(\bQ_1)$, interesting logics from the point
of view of complexity necessarily include quantifiers of all arities.

One way of getting meaningful logics that include quantifiers of
unbounded arity is to consider quantifiers restricted to stronger closure conditions than just closure under isomorphisms.  In recent work, novel game-based methods have established new inexpresibilty results for such logics, i.e.\ logics with a wide class of quantifiers of unbounded arity, but satisfying further restrictions.
An important example is the class of
linear-algebraic quantifiers, introduced in~\cite{DawarGP19} which is the closure under interpretations of binary quantifiers invariant under invertible linear maps over finite fields.  Equivalence in the resulting logic is characterized by the invertible map games introduced
in~\cite{DawarH17}.  These games are used in a highly sophisticated
way by Lichter~\cite{Lichter21} to demonstrate a polynomial-time
property that is not definable in fixed-point logic with rank
introduced in~\cite{DawarGHL09,GraedelP19}.  The
result is extended to the infinitary logic with all linear-algebraic
quantifiers in~\cite{DawarGL22}.

Another example is the recent result of Hella~\cite{Hella23} showing a
hierarchy theorem for quantifiers based on \emph{constraint satisfaction problems} (CSP), using a novel game.  Recall that
for a fixed relational structure $\mB$, $\CSP(\mB)$ denotes the class
of structures that map homomorphically to $\mB$.  Hella establishes
that, for each $n > 1$, there is a structure $\mB$ with $n+1$ elements
that is not definable in $L^{\omega}_{\infty\omega}(\bQ_1,\bCSP_n)$,
where $\bCSP_n$ denotes the collection of all quantifiers of the form
$Q_{\CSP(\mB')}$ where $\mB'$ has at most $n$ elements.  Note that
$\bCSP_n$ includes quantifiers of all arities.  

The interest in CSP quantifiers is inspired by the great progress that
has been made in classifying constraint satisfaction problems in
recent years, resulting in the dichotomy theorem of Bulatov and
Zhuk~\cite{Bulatov17,Zhuk20}.  The so-called algebraic approach to the
classification of CSP has shown that the complexity of $\CSP(\mB)$ is
completely determined by the algebra of polymorphisms of the structure $\mB$.  In particular, the complexity is completely determined by the equational theory of this algebra.  As we make explicit in Section~\ref{sec:PartialPoly} below, equations satisfied by the polymorphisms of $\mB$ naturally give rise to certain closure properties for the class of structures $\CSP(\mB)$, which we describe by \emph{partial polymorphisms}.

A central aim of the present paper is to initiate the study of quantifiers closed under partial polymorphisms.  We present a Spoiler-Duplicator pebble game, based on bijection games, which exactly characterises the expressive power of such quantifiers.  More precisely, there is such a game for any suitable family $\cP$ of partial polymorphisms.  The exact definition of the game and the proof of the characterization are given in Section~\ref{Games}.

As a case study, we consider the partial polymorphisms described by a
\emph{near-unanimity} condition.  It is known since the seminal work
of Feder and Vardi~\cite{FederV98} that if a structure $\mB$ admits a
near-unanimity polymorphism, then $\CSP(\mB)$ has \emph{bounded
  width}, i.e.\ it (or more precisely, its complement) is definable in
Datalog.  On the other hand, the problem of determining the
solvability of a system of equations over the two-element field
$\Zp{2}$ is the classic example of a tractable CSP that is not of
bounded width.  Indeed, it is not even definable in
$C^{\omega}_{\infty\omega}$~\cite{AtseriasBD09}.  We show that the
collection of quantifiers that are closed under near-unanimity partial
polymorphisms is much richer than the classes  $\CSP(\mB)$ where $\mB$
has a near-unanimity polymorphism.  The collection not only includes
quantifiers which are not CSP, but it also includes CSP quantifiers
which are not of bounded width, including intractable ones such as
hypergraph colourability.  Still, we are able to show that the problem
of solving systems of equations over $\Zp{2}$ is not definable in the
extension of $C^{\omega}_{\infty\omega}$ with \emph{all} quantifiers
closed under near-unanimity partial polymorphisms.  This sheds new
light on the inter-definability of constraint satisfaction problems.
For instance, while it follows from the arity hierarchy
of~\cite{Hella96} that the extension of $C^{\omega}_{\infty\omega}$
with a quantifier for graph 3-colourability still cannot define
solvability of systems of equations over $\Zp{2}$, our result shows
this also for the extension of $C^{\omega}_{\infty\omega}$ with all
hypergraph colourability quantifiers.

\section{Preliminaries}\label{sec:Prelim}

We assume basic familiarity with logic, and in particular the logics
commonly used in finite model theory (see~\cite{EbbinghausF99}, for
example).  We write $L^k_{\infty\omega}$ to denote the infinitary logic  (that is, the closure of first-order logic with infinitary conjunctions and disjunctions) with $k$ variables and $L^{\omega}_{\infty\omega}$ for $\bigcup_{k \in \omega} L^k_{\infty\omega}$.
We are mainly interested in the extensions of these logics with generalized quantifiers, which we introduce in more detail in Section~\ref{GQ} below.

We use Fraktur letters $\mA, \mB, \ldots$ to denote structures and the corresponding Roman letters $A, B, \ldots$ to denote their universes.
Unless otherwise mentioned, all structures are assumed to be finite.
We use function notation, e.g.\ $f: A \ra B$ to denote possibly
\emph{partial} functions.  If $f:A \ra B$ is a function and $\vec a
\in A^m$ a tuple, we write $f(\vec a)$ for the tuple in $B^m$ obtained
by applying $f$ to $\vec a$ componentwise.  If $\vec a_1,\ldots,\vec
a_n$ is a sequnce of $m$-tuples, write $(\vec a_1,\ldots,\vec
a_n)^T$ for the sequence $\vec b_1,\ldots, \vec b_m$ of $n$-tuples,
where $\vec b_i$ is the tuple of $i$th components of $\vec a_1,\ldots,\vec
a_n$.  Given a function $f: A^n \ra B$, we write $\hat f(\vec a_1,\ldots,\vec
a_n)$ to denote $f((\vec a_1,\ldots,\vec
a_n)^T) = (f(\vec b_1),\ldots,f(\vec b_m)) $.

For a pair of structures $\mA$ and $\mB$, a \emph{partial isomorphism}
from $\mA$ to $\mB$ is a partial function $f: A \ra B$ which is an
isomorphism between the substructure of $\mA$ induced by the domain of
$f$ and the substructure of $\mB$ induced by the image of $f$.  We
write $\PI(\mA,\mB)$ to denote the collection of all partial
isomorphisms from $\mA$ to $\mB$.

We write $\N$ or $\omega$ to denote the natural numbers, and $\Z$ to denote the ring of integers.  For any $n \in \N$, we write $[n]$ to denote the set $\{1,\ldots,n\}$.  When mentioned without further qualification, a graph $G = (V,E)$ is simple and undirected.  That is, it is a structure with universe $V$ and one binary relation $E$ that is irreflexive and symmetric.  The \emph{girth} of a graph $G$ is the length of the shortest cycle in $G$.

\subsection{Generalized quantifiers}\label{GQ}

Let $\sigma, \tau$ be relational vocabularies 
with $\tau = \{R_1,\ldots,R_m\}$, and $\ar(R_i)=r_i$
for each $i\in [m]$.
An interpretation $\mathcal{I}$ of $\tau$ in $\sigma$ with parameters
$\vec{z}$ is a tuple of $\sigma$-formulas $(\psi_1,\ldots,\psi_m)$
along with tuples $\vec y_1,\ldots,\vec y_m$ of variables with $|\vec
y_i|=r_i$ for $i\in [m]$, such that the free variables of $\psi_i$ are
among $\vec{y}_i \vec z$.  Such an interpretation defines a mapping
that takes a $\sigma$-structure $\mA$, along with an interpretation
$\alpha$ of the parameters $\vec{z}$ in $\mA$ to a $\tau$-structure
$\mB$ as follows.   The universe of $\mB$ is $A$, and the relations
$R_i \in \tau$ are interpreted in $\mB$ by $R_i^{\mB} = \{\vec b \in
A^{r_i} \mid (\mA, \alpha[\vec b/\vec y_i])
\models \psi_i\}$.

Let $L$ be a logic and $\cK$ a class of $\tau$-structures.  The extension $L(Q_{\cK})$ of $L$ by the \emph{generalized quantifier} for
the class $\cK$ is obtained by extending the syntax of $L$ by the
following formula formation rule:
\begin{quote}
	For $\cI = (\psi_{1},\ldots,\psi_{m})$ an interpretation of $\tau$ in $\sigma$ with parameters $\vec{z}$, $\psi(\vec{z}) = Q_{\cK}\vec{y}_1,\ldots,\vec{y}_m \cI$ is a formula over the signature $\sigma$, with free variables $\vec z$.  The semantics of the formula is given by
$(\mA,\alpha) \models \psi(\vec{z})$, if, and only
if,  $\mB := \cI(\mA,\alpha)$ is  in the class $\cK$.
\end{quote}

The extension $L(\bQ)$ of $L$ by a collection $\bQ$ of generalized quantifiers
is defined by adding the rules above to $L$ for each $Q_\cK\in\mathbf{Q}$ separately.

The \emph{type} of the quantifier $Q_\cK$ is $(r_1,\ldots,r_m)$, and the
\emph{arity} of $Q_\cK$ is $\max\{r_1,\ldots,r_m\}$.
For the sake of simplicity, we assume in the sequel that the type of $Q_\cK$ 
is \emph{uniform}, i.e., $r_i=r_j$ for all $i,j\in [m]$. 
This is no loss of generality, since any quantifier $Q_\cK$ is definably equivalent with
another quantifier $Q_{\cK'}$ of uniform type with the same arity.
Furthermore, we restrict the syntactic rule of $Q_\cK$ by requiring that $\vec y_i=\vec y_j$
for all $i,j\in [m]$. Then we can denote the formula obtained by applying the rule simply
by $\varphi=Q_\cK \vec y\,(\psi_1,\ldots,\psi_m)$. 
Note however, that this convention disallows formulas of the type $\theta=Q x,y\,(R(x,y),R(y,x))$ in which both $x$ and $y$ remain free even though $x$ is bound in $R(x,y)$ and $y$ is bound in $R(y,x)$, and hence weakens the expressive power of $\FO^k(Q_\cK)$ and $L^k_{\infty\omega}(Q_\cK)$. Fortunately the loss can be compensated by using more variables (e.g., $\theta$ is equivalent with $Q z\,(R(z,y),R(z,x))$, whence the restriction does not affect the expressive power of $\FO(Q_\cK)$ and $L^\omega_{\infty\omega}(Q_\cK)$.

Let $Q=Q_\cK$ and $Q'=Q_{\cK'}$ be generalized quantifiers. We say that $Q$ is 
\emph{definable} in $L(Q')$ if the defining class $\cK$ is definable in $L(Q')$, i.e.,
there is a sentence $\varphi$ of $L(Q')$ such that $\cK=\{\mA\mid \mA\models\varphi\}$.

We write $\bQ_n$ to denote the collection of all quantifiers of arity at most $n$.  Hella~\cite{Hella96} shows that for any $n$, there is a quantifier of arity $n+1$ that is not definable in $L^{\omega}_{\infty\omega}(\bQ_n)$.  The logic $L^{\omega}_{\infty\omega}(\bQ_1)$ is equivalent to $C^{\omega}_{\infty\omega}$, the infinitary logic with counting.
The notion of interpretation we have defined is fairly
restricted in that it does not allow for \emph{relativization} or
\emph{vectorizations} (see, e.g.~\cite[Def.~12.3.6]{EbbinghausF99}.  The relativizations and vectorizations of a quantifer $Q$ can always be seen as a \emph{collection} of simple quantifiers of unbounded arity.

\subsection{CSP and polymorphisms}
Given relational structures $\mA$ and $\mB$ over the same vocabulary $\tau$, a \emph{homomorphism} $h: \mA \ra \mB$ is a function that takes elements of $A$ to elements of $B$ and such that for every $R \in \tau$ of arity $r$ and any $\vec a \in A^r$, $\vec a \in R^{\mA}$ implies $h(\vec a) \in R^{\mB}$.  For a fixed structure $\mB$, we write $\CSP(\mB)$ to denote the collection of structures $\mA$ for which there is some homomorphism $h: \mA \ra \mB$.  By the celebrated theorem of Bulatov and Zhuk, every class  $\CSP(\mB)$ is either decidable in polynomial time or NP-complete.

Given  a $\tau$-structure $\mB$ and $m \in \N$, we define a $\tau$-structure $\mB^m$. Its universe is $B^m$ and if  $R$ in $\tau$ is a relation of arity $r$, and $\vec{a}_i = (a^1_i,\ldots,a^m_i)$ is an $m$-tuple of elements of $B$, for each $i \in [r]$, then $(\vec{a}_1,\ldots,\vec{a}_r) \in R^{\mB^m}$ if, and only if, for each $j \in [m]$, $(a^j_1,\ldots,a^j_r) \in R^{\mB}$.  Then, a \emph{polymorphism} of $\mB$ is a homomorphism $p: \mB^m \ra \mB$ for some $m$.  The collection of polymorphisms of $\mB$ forms an algebraic \emph{clone} with universe $B$.  It is known that the equational theory of this algebra completely determines the computational complexity of $\CSP(\mB)$ (see~\cite{BartoKW17} for an expository account).

A function $m: B^3 \ra B$ is a \emph{majority} function if it satisfies the equations $m(a,a,b) = m(a,b,a) = m(b,a,a) = a$ for all $a,b \in B$.   More generally, for $\ell \geq 3$, a function $n: B^{\ell} \ra B$ is a \emph{near-unanimity} function of arity $\ell$ if for any $\ell$-tuple $\vec a$, we have $n(\vec a) = a$ whenever at least $\ell -1$ components of $\vec a$ are $a$.  In particular, a near-unanimity function of arity $3$ is a majority function.  A function $M: B^3 \ra B$ is a \emph{Maltsev} function if it satisfies the identities $M(a,b,b) = M(b,b,a) = a$ for all $a,b \in B$.

For any structure $\mB$ which has a near-unanimity polymorphism, the
class $\CSP(\mB)$ is decidable in polynomial time, and definable in
$L^{\omega}_{\infty\omega}$.  If $\mB$ admits a Maltsev polymorphism,
then $\CSP(\mB)$ is also decidable in polynomial time, but may not be
definable in $L^{\omega}_{\infty\omega}$ or
$L^{\omega}_{\infty\omega}(\bQ_1)$, its extension with all unary
quantifiers.  The classic example of a CSP with a Maltsev polymorphism
that is not definable in $L^{\omega}_{\infty\omega}(\bQ_1)$ is solvability of systems of equations over $\Zp{2}$ with $\ell$ variables per equation.  We can treat this as the class of structures $\CSP(\mC_{\ell})$ where $\mC_{\ell}$ is the structure with universe $\{0,1\}$ and two $\ell$-ary relations $R_0 = \{(b_1,\ldots,b_{\ell}) \mid \sum_i b_i \equiv 0 \pmod 2\}$ and $R_1 = \{(b_1,\ldots,b_{\ell}) \mid \sum_i b_i \equiv 1 \pmod 2\}$.

If $\cK = \CSP(\mB)$ for some fixed structure $\mB$, we call $Q_{\cK}$
a \emph{CSP quantifier}.  Write $\mathbf{CSP}_n$ for the collection of
all CSP quantifiers $Q_{\cK}$ where $\cK = \CSP(\mB)$ for a structure
with at most $n$ elements.   Note that $\mathbf{CSP}_n$ contains
quantifiers of all arities.   Hella~\cite{Hella23} defines a pebble
game that characterizes equivalence of structures in the logic
$L^{\omega}_{\infty\omega}(\bQ_1,\mathbf{CSP}_n)$ and shows that there
is a structure $\mB$ on $n+1$ elements such that $\CSP(\mB)$ is not definable in this logic.

\section{Partial polymorphisms}\label{sec:PartialPoly}

Let $\tau$ be a relational vocabulary, and let $\mC$ be a
$\tau$-structure with a polymorphism $p\colon \mC^n\to\mC$.  This
gives rise to a closure condition on the class $\CSP(\mC)$.
In particular, suppose $\mB \in \CSP(\mC)$ by a
homomorphism $h: \mB \ra \mC$.  We can, in a sense, ``close''  $\mB$
under the polymorphism $p$ by including in each relation $R^{\mB}$ ($R
\in \tau$) any
tuple $\vec a$ for which $h(\vec a) = p(h(\vec a_1,\ldots,\vec a_n))$
for some $\vec a_1,\ldots,\vec a_n \in R_i^{\mB}$.  The resulting
structure $\mB'$ is still in $\CSP(\mC)$ as is any structure $\mA$
with the same universe as $\mB$ and for which $R^{\mA} \subseteq
R^{\mB'}$ for all $R \in \tau$.

Our aim is to generalize this type of closure properties from $\csp$ quantifiers to a 
larger class of generalized quantifiers.  To formally define this, it
is useful to introduce some notation.  For reasons that will become
clear, we use \emph{partial} functions $p$.

\begin{definition}
Let $A\not=\emptyset$ be a set, and let $p$ be a  be a partial function $A^n\to A$. 

(a) If $R\subseteq A^r$, then $p(R):= \{\hat p(\vec a_1,\ldots,\vec a_n)\mid \vec a_1,\ldots,\vec a_n\in R\}$.

(b) If $\mA=(A,R_1^\mA,\ldots,R_m^\mA)$, then
we denote the structure $(A,p(R_1^\mA),\ldots,p(R_m^\mA))$ by $p(\mA)$.
\end{definition}

We say that $p$ is a \emph{partial polymorphism} of a $\tau$-structure $\mA$ with domain $A$ if for every $R\in\tau$, the relation $R^\mA$ is closed with respect to $p$, i.e.,
$p(R^\mA)\subseteq R^\mA$.

The reason for considering partial functions is that we are usually
interested in polymorphisms that satisfy certain equations.  The
equations specify the polymorphism partially, but not totally.  We can
uniformly specify closure properties on our class of structures for
all polymorphisms satisfying the equations by only requiring closure
for the common partial function.  This is illustrated in the examples below.

By a \emph{family of partial functions} we mean a class $\cP$ that contains a partial function
$p_A\colon A^n\to A$ for every finite set $A$, where $n$ is a fixed positive integer. We give next
some important examples of families of partial functions that arise naturally from well-known
classes of polymorphisms.

\begin{example}\label{partial-polym}
(a) The \emph{Maltsev family} $\cM$ consists of the partial functions $M_A\colon A^3\to A$
such that $M_A(a,b,b)=M_A(b,b,a)=a$ for all $a,b\in A$, and $M_A(a,b,c)$ is undefined unless $a=b$ or $b=c$. If $\mA$ has a Maltsev polymorphism $p\colon A^3\to A$, then clearly $M_A$ is a restriction of $p$, whence it is a partial polymorphism of $\mA$.

(b) The family $\cMaj$ of ternary \emph{partial majority functions}
consists of the  partial functions $m_A\colon A^3\to A$ such that
$m_A(a,a,b)=m_A(a,b,a)=m_A(b,a,a)=a$ for all $a,b\in A$, and
$m_A(a,b,c)$ is undefined if $a,b$ and $c$ are all distinct. If $\mA$ has a majority polymorphism, then $m_A$ is a restriction of it, whence it is a partial polymorphism of $\mA$.

(c) More generally, for each $\ell\ge 3$ we define the family $\cN_\ell$ of \emph{$\ell$-ary partial near-unanimity functions} $n^\ell_A\colon A^\ell\to A$ as follows:
\begin{itemize}
\item $n^\ell_A(a_1,\ldots,a_\ell)=a$ if and only if $|\{i\in [n]\mid a_i=a\}|\ge \ell-1$.
\end{itemize} 
In particular, $\cMaj=\cN_3$.
\end{example}

We next give a formal definition for the closure property of generalized quantifiers 
that arises from a family of partial functions. 
In the definition we use the notation $\mA\le\mB$ if $\mA$ and $\mB$ are 
$\tau$-structures such that $A=B$ and $R^\mA\subseteq R^\mB$ for each $R\in\tau$.
Furthermore, we define the union $\mA\cup\mB$ of $\mA$ and $\mB$ to be the $\tau$-structure 
$\mC$ such that $C=A\cup B$ and $R^\mC=R^\mA\cup R^\mB$ for each $R\in\tau$.

\begin{definition}\label{P-closed}
Let $\cP$ be a family of $n$-ary partial functions, and let $Q_\cK$ be a generalized quantifier of vocabulary $\tau$. We say that $Q_\cK$ is $\cP$-closed if the following holds for all $\tau$-tructures $\mA$ and $\mB$ with $A=B$:
\begin{itemize}
\item  if $\mB\in\cK$ and $\mA\le p_A(\mB)\cup\mB$, then $\mA\in\cK$.
\end{itemize}
We denote the class of all $\cP$-closed quantifiers by $\bQ_\cP$.
\end{definition}

Note that the condition $\mA\le p_A(\mB)\cup\mB$ holds if and only if for every $R\in\tau$ and every $\vec a\in R^\mA\setminus R^\mB$ there are tuples $\vec a_1,\ldots,\vec a_n\in R^\mB$ such that 
$\vec a=\widehat{p_A}(\vec a_1,\ldots,\vec a_n)$.

The quantifier $Q_\cK$ is \emph{downwards monotone}, if $\mA\le\mB$ and $\mB\in\cK$ implies $\mA\in\cK$. It follows directly from Definition~\ref{P-closed} that all $\cP$-closed quantifiers are downwards monotone.

\begin{proposition}
If  $Q_\cK\in\bQ_\cP$ for some family $\cP$, then $Q_\cK$ is 
downwards monotone.
\end{proposition}

It is easy to see that, for any family $\cP$, the first-order quantifiers can be defined from
a $\cP$-closed quantifier using only negation.

\begin{proposition}\label{exists-P-closed}
Let $\cK_0$ be the class of all $\{P\}$-structures $\mA$ such that $P^\mA=\emptyset$. Then $Q_{\cK_0}\in\bQ_\cP$ for any family $\cP$ of partial functions.
\end{proposition}

\begin{proof}
If $\mB\in\cK_0$, then $P^\mB=\emptyset$, whence $p_B(\mB)=\emptyset$. Thus, if $\mA\le p_B(\mB)\cup\mB$, then $P^\mA=\emptyset$, and hence $\mA\in\cK_0$.
\end{proof}

Note that in the case $\ar(P)=1$, the quantifier $Q_{\cK_0}$ of the proposition above is the negation 
of the existential quantifier: $\mA\models Q_{\cK_0} x\,\varphi\iff\mA\models\lnot\exists x\,\varphi$.

Up to now we have not imposed any restrictions on the family $\cP$. 
It is natural to require that the partial functions in $\cP$ are uniformly defined, or at least
that $(A,p_A)$ and $(B,p_B)$ are isomorphic if $|A|=|B|$. Such requirements are captured 
by the notions defined below.

\begin{definition}\label{p-f-families}
Let $\cP$ be a family of $n$-ary partial functions.

(a) $\cP$ is \emph{invariant} if it respects bijections: 
if $f\colon A\to B$ is a bijection and $a_1,\ldots,a_n\in A$, then $p_B(f(a_1),\ldots,f(a_n))\simeq f(p_A(a_1,\ldots,a_n))$. Here the symbol $\simeq$ says that either both sides are defined and
have the same value, or both sides are undefined.

(b) $\cP$ is \emph{strongly invariant} if it respects injections: 
if $f\colon A\to B$ is an injection and $a_1,\ldots,a_n\in A$, then $p_B(f(a_1),\ldots,f(a_n))\simeq f(p_A(a_1,\ldots,a_n))$.

(c) $\cP$ is \emph{projective}, if it strongly invariant and it is preserved by all functions:
if $f\colon A\to B$ is a function and $a_1,\ldots,a_n\in A$ are such that $p_A(a_1,\ldots,a_n)$ 
is defined, then $p_B(f(a_1),\ldots,f(a_n))= f(p_A(a_1,\ldots,a_n))$.
\end{definition}

It is easy to verify that $\cP$ is invariant if, and only if, it is
determined by equality types on each cardinality: there are quantifier
free formulas in the language of equality
$\theta^m_\cP(\vec x,y)$ 
such that if $|A|=m$, then $p_A(\vec a)=b\iff A\models\theta^m_\cP[\vec a/\vec x,b/y]$ 
holds for all $\vec a\in A^n$ and $b\in A$. 
Similarly, $\cP$ is strongly invariant if, and only if, the same holds with a single formula
$\theta_\cP=\theta^m_\cP$ for all $m\in\omega$. 

Note that if the family $\cP$ is strongly invariant, then for every finite set $A$, $p_A$ is 
a \emph{partial choice function}, i.e., $p_A(a_1,\ldots,a_n)\in\{a_1,\ldots,a_n\}$. Indeed,
if $b:=p_A(a_1,\ldots,a_n)\not\in\{a_1,\ldots,a_n\}$ and $B=A\cup\{c\}$, where $c\notin A$,
then using the identity function $f=\mathrm{id}_A$ of $A$ in the condition $p_B(f(a_1),\ldots,f(a_n))= f(p_A(a_1,\ldots,a_n))$,
we get 
$p_B(a_1,\ldots,a_n)= b$. On the other hand, using the injection $f'\colon A\to B$ that
agrees with $\mathrm{id}_A$ on $A\setminus\{b\}$ but maps $b$ to $c$, we get
the contradiction $p_B(a_1,\ldots,a_n)= c\not=b$. 

\begin{remark}
An invariant family may contain functions $p_A$ that are not partial choice functions: for example the family consisting of all functions $p_A\colon A^n\to A$ such that $p_A(a_1,\ldots,a_n)=a_{n+1}\iff A\setminus\{a_1,\ldots,a_n\}=\{a_{n+1}\}$ is invariant. However, if $|A|>n+1$, then $p_A$ is necessarily a partial choice function.
\end{remark}

\begin{lemma}\label{unary-P-closed}
Let $\cP$ be a family of $n$-ary partial choice functions.
Then $Q_\cK\in\bQ_\cP$ for any unary downwards monotone quantifier $Q_\cK$.
In particular this holds if $\cP$ is strongly invariant.
\end{lemma}

\begin{proof}
Let $\tau$ be the vocabulary of $\cK$, and assume that $\mB\in\cK$ and $\mA\le p_A(\mB)\cup\mB$. Then for all $R\in\tau$ and $a\in R^\mA\setminus R^\mB$ there are $a_1,\ldots,a_n\in A$ such that $p_A(a_1,\ldots,a_n)=a$ and $a_i\in R^\mB$ for each $i\in [n]$. Since $p_A$ is a choice function, we have $a\in\{a_1,\ldots,a_n\}$, and hence $a\in R^\mB$. Thus we see that $\mA\le\mB$, and consequently $\mA\in\cK$, since $Q_\cK$ is downwards monotone.
\end{proof}


It is easy to see that the families $\cM$ and $\cN_\ell$,
$\ell\ge 3$, introduced in Example~\ref{partial-polym}, are strongly
invariant.  Indeed, the defining formulas $\theta_{\cM}$ and
$\theta_{\cN_{\ell}}$ are easily obtained from the identities that
define these conditions.
Thus, all unary downwards monotone quantifiers are
$\cM$-closed and $\cN_\ell$-closed. For the families $\cN_\ell$ we can prove a much
stronger result:

\begin{lemma}\label{r-N-ell-closed}
Let $\ell\ge 3$, and let $Q_\cK$ be a downwards monotone quantifier of arity $r<\ell$. Then
$Q_\cK\in\bQ_{\cN_\ell}$.
\end{lemma}

\begin{proof}
Let $\tau$ be the vocabulary of $\cK$, and assume that $\mB\in\cK$ and 
$\mA\le n^\ell_A(\mB)\cup\mB$. 
Then for all $R\in\tau$ and $\vec a=(a_1,\ldots,a_r)\in R^\mA\setminus R^\mB$ there are 
$\vec a_i=(a^1_i,\ldots,a^r_i)\in R^\mB$, $i\in [\ell]$, such that 
$\widehat{n^\ell_A}(\vec a_1,\ldots,\vec a_\ell)=\vec a$.
Thus, for each $j\in [r]$ there is at most one $i\in [\ell]$ such that $a^j_i\not= a_j$,
and hence there is at least one $i\in [\ell]$ such that $\vec a=\vec a_i$. This shows that
$\mA\le\mB$, and since $Q_\cK$ is downwards monotone, we conclude that $\mA\in\cK$.
\end{proof}

Using a technique originally due to Imhof  for (upwards) monotone quantifiers (see \cite{HellaI98}), 
we can show that any quantifier $Q_\cK$ is definable by a
downwards monotone quantifier of the same arity. Indeed, if the vocabulary of $\cK$
is $\tau=\{R_1,\ldots,R_m\}$, where $\ar(R_i)=r$ for all $i\in [m]$, we let 
$\tau':=\{S_1,\ldots,S_m\}$ be a disjoint copy of $\tau$, and 
$\tau^*:=\tau\cup\tau'$. Furthermore, we let $\cK^*$ be the class of all $\tau^*$-structures
$\mA$ such that $R^\mA_i\cap S^\mA_i=\emptyset$ for all $i\in [m]$, and 
 $(A,R^\mA_1,\ldots,R^\mA_m)\in\cK$ or $R^\mA_i\cup S^\mA_i\not=A^r$
for some $i\in [m]$. Then $Q_{\cK^*}$ is downwards monotone, and clearly
$Q_\cK\vec x\,(\psi_1,\ldots,\psi_m)$ is equivalent with
$Q_{\cK^*}\vec x\,(\psi_1,\ldots,\psi_m,\lnot\psi_1,\ldots,\lnot\psi_m)$.

Using this observation, we get the following corollary to Lemmas~\ref{unary-P-closed} 
and \ref{r-N-ell-closed}.

\begin{corollary}\label{less-than-ell-ary}
(a) Let $\cP$ be as in Lemma~\ref{unary-P-closed}. Then 
$L^k_{\infty\omega}(\bQ_{\cP}\cup\bQ_1)\le L^k_{\infty\omega}(\bQ_{\cP})$.

(b) $L^k_{\infty\omega}(\bQ_{\cN_\ell}\cup\bQ_{\ell-1})\le L^k_{\infty\omega}(\bQ_{\cN_\ell})$.
\end{corollary}

As explained in the beginning of this section, the definition of $\cP$-closed quantifiers
was inspired by the closure property of a $\csp$ quantifier $Q_{\CSP(\mC)}$ that 
arises from a polymorphism of $\mC$. Thus, it is natural to look for sufficient conditions
on the family $\cP$ and the target structure $\mC$ for $Q_{\CSP(\mC)}$
to be
$\cP$-closed. It turns out that the notions of projectivity and partial polymorphism
lead to such a condition.

\begin{proposition}\label{projectiveCSP}
Let $\cP$ be a projective family of $n$-ary partial functions, and let
$\mC$ be a $\tau$-structure. If $p_C$ is a partial polymorphism of
$\mC$, then $Q_{\CSP(\mC)}\in\bQ_\cP$. 
\end{proposition}

\begin{proof}
Assume that $\mB\in\cK$ and $\mA\le p_A(\mB)\cup\mB$. Then $A=B$ and there is a homomorphism
$h\colon \mB\to \mC$. We show that $h$ is a homomorphism $\mA\to\mC$, and hence
$\mA\in\cK$. Thus let $R\in\tau$, and let $\vec a\in R^\mA$. If $\vec a\in R^\mB$, then $h(\vec a)\in R^\mC$ by assumption. On the other hand, if $\vec a\in R^\mA\setminus R^\mB$, then there exist
tuples $\vec a_1,\ldots,\vec a_n\in R^\mB$ such that  $\vec a=\widehat{p_A}(\vec a_1,\ldots,\vec a_n)$.
Since $h$ is a homomorphism $\mB\to\mC$, we have $h(\vec a_i)\in R^\mC$ for each $i\in [n]$.
Since $p_C$ is a partial polymorphism of $\mC$, we have 
$\widehat{p_C}(h(\vec a_1),\ldots,h(\vec a_n))\in R^\mC$. Finally, since $\cP$ is projective,
we have $h(\vec a)=h(\widehat{p_A}(\vec a_1,\ldots,\vec a_n))=
\widehat{p_C}(h(\vec a_1),\ldots,h(\vec a_n))$, and hence $h(\vec a)\in R^\mC$.
\end{proof}

We can now apply Proposition~\ref{projectiveCSP}
to the families introduced in Example~\ref{partial-polym}.

\begin{example}\label{P-closed-CSP}
(a) Consider a constraint satisfaction problem $\CSP(\mC)$ such that $\mC$ has a Maltsev polymorphism $p\colon \mC^3\to \mC$. We show that $Q_{\CSP(\mC)}\in\bQ_\cM$. As pointed out in Example~\ref{partial-polym}, $M_C$ is a partial polymorphism of $\mC$. Thus, by Proposition~\ref{projectiveCSP} it suffices to show that the Maltsev family $\cM$ is projective.

Thus, assume that $f\colon A\to B$ is a function, and $M_A(a,b,c)$ is defined. Then $a=b$ and $M_A(a,b,c)=c$, or $b=c$ and $M_A(a,b,c)=a$. In the former case we have $f(a)=f(b)$, whence $M_B(f(a),f(b),f(c))=f(c)=f(M_A(a,b,c))$. In the latter case we have $f(b)=f(c)$, whence $M_B(f(a),f(b),f(c))=f(a)=f(M_A(a,b,c))$.

(b) The \emph{$n$-regular hypergraph $m$-colouring} problem is $\CSP(\mH_{n,m})$, where $\mH_{n,m}=([m],R_{n,m})$ is the complete $n$-regular hypergraph with $m$ vertices, i.e., 
\begin{itemize}
\item $R_{n,m}:=\{(v_1,\ldots,v_n)\in [m]^n\mid v_i\not=v_j\text{ for all }1\le i<j\le m\}$. 
\end{itemize}
We show that $Q_{\CSP(\mH_{n,m})}\in\bQ_\cMaj$ for all $n\ge 2$ and $m\ge n$. By Proposition~\ref{projectiveCSP} it suffices to show that $m_{[m]}$ is a partial polymorphism of $\mH_{n,m}$, and the family $\cMaj$ is projective.

To see that $m_{[m]}$ is a partial polymorphism of $\mH_{n,m}$, assume that $\vec a_i=(a_i^1,\ldots,a_i^n)\in R_{n,m}$ for $i\in [3]$, and $\vec a=(a_1,\ldots,a_n)=\hat{m_{[m]}}(\vec a_1,\vec a_2,\vec a_3)$. By the definition of $m_{[m]}$, for each $j\in [n]$ we have $|\{i\in [3]\mid a^j_i=a_j\}|\ge 2$. Thus for any two distinct $j,k\in [n]$, there is $i\in [3]$ such that $a_j=a^j_i$ and $a^k_i=a_k$, whence $a_j\not=a_k$.  Thus we have $\vec a\in R_{n,m}$.

To show that $\cMaj$ is projective, assume that $f\colon A\to B$ is a function, and $m_A(a,b,c)$ is defined. Then $a=b=m_A(a,b,c)$, $a=c=m_A(a,b,c)$ or $b=c=m_A(a,b,c)$. In the first case we have $f(m_A(a,b,c))=f(a)=f(b)= m_B(f(a),f(b),f(c))$, as desired. The two other cases are similar. 

(c) In the same way we can show that the family $\cN_\ell$ of partial near-unanimity polymorphisms is projective for any $\ell\ge 3$.  We relax now the notion of hypergraph coloring as follows: Let $\mH=(H,R)$, where $R\subseteq H^{[n]}$, be a hypergraph and let $k<n$. A \emph{$k$-weak $m$-coloring} of $\mH$ is a function $f\colon H\to [m]$ such that for all $e\in R$ and all $i\in [m]$, $|e\cap f^{-1}[\{i\}]|\le k$. Observe now that there exists a $k$-weak $m$-coloring of $\mH$ if and only if $\mH\in\CSP(\mH^k_{n,m})$, where $\mH^k_{n,m}=([m],R^k_{n,m})$ is the structure such that
\begin{itemize}
\item $R^k_{n,m}:=\{(v_1,\ldots,v_n)\in [m]^n\mid |\{v_i\mid i\in I\}|\ge 2 \text{ for all }I\subseteq [n]\text{ with }|I|= k+1\}$. 
\end{itemize}
Note that $\mH^1_{n,m}=\mH_{n,m}$, whence $m_{[m]}=n^3_{[m]}$ is a partial polymorphism of
$\mH^1_{n,m}$. It is straightforward to generalize this to $\ell>3$:
$n^\ell_{[m]}$ is a partial polymorhism of $\mH^{\ell-2}_{n,m}$. Thus
by Proposition~\ref{projectiveCSP}, the $\csp$ quantifier
$Q_{\CSP(\mH^{\ell-2}_{n,m})}$ is $\bQ_{\cN_\ell}$-closed.
\end{example}

\begin{remark}
As shown in Example~\ref{P-closed-CSP}(b), the partial majority function $m_{[m]}$ 
is a partial polymorphism of the structure $\mH_{n,m}$. However, there does not exist
any polymorphism $p\colon [m]^3\to [m]$ that extends $m_{[m]}$. This can be verified
directly, but it also follows from the fact that $\CSP(\mC)$ is of
bounded width for any $\mC$ that has a majority polymorphism (\cite{FederV98}), but
$\CSP(\mH_{n,m})$ is not of bounded width.
The same holds for the partial functions $n^\ell_{[m]}$ and the structures $\mH^k_{n,m}$ 
in Example~\ref{P-closed-CSP}(c).
\end{remark}

\section{Pebble game for $\cP$-closed quantifiers}\label{Games}

In this section we introduce a pebble game that characterizes equivalence of structures with respect to  $L^\omega_{\infty\omega}(\bQ_\cP)$, the extension of the infinitary $k$-variable logic $L^\omega_{\infty\omega}$ by the class of all $\cP$-closed quantifiers.

We fix a family 
$\cP$ of $n$-ary partial functions for the rest of the section.
Given two structures $\mA$ and $\mB$ of the same vocabulary, and assignments
$\alpha$ and $\beta$ on $\mA$ and $\mB$, respectively, such that
$\dom(\alpha)=\dom(\beta)$, we write
$(\mA,\alpha)\equiv^k_{\infty\omega,\cP}(\mB,\beta)$ if the equivalence
\begin{itemize}
\item[] $(\mA,\alpha)\models\varphi\iff (\mB,\beta)\models\varphi$
\end{itemize}
holds for all formulas $\varphi\in L^k_{\infty\omega}(\bQ_\cP)$ with free
variables in $\dom(\alpha)$. 
If $\alpha=\beta=\emptyset$, we write simply $\mA\equiv^k_{\infty\omega,\cP}\mB$ instead of
$(\mA,\emptyset)\equiv^k_{\infty\omega,\cP}(\mB,\emptyset)$. 

The basic idea of our pebble game for a pair $(\mA,\mB)$ 
of structures is the following. In each round  Duplicator gives a
bijection $f\colon A\to B$, just like in the bijection games of
\cite{Hella96}, but instead of using $\vec b=f(\vec a)$ as answer for
Spoiler's move $\vec a\in A^r$, she is allowed to give a sequence
$\vec b_1,\ldots,\vec b_n\in B^r$ of alternative answers as long as
$\vec b=\widehat{p_B}(\vec b_1,\ldots,\vec b_n)$. Spoiler completes the round by
choosing one of these alternatives $\vec b_i$. Spoiler wins if $\vec
a\mapsto\vec b_i$ is not a partial isomorphism; otherwise the game
carries on from the new position.

Observe now that if Duplicator 
has a winning strategy for the first round of the game, then $f(\mA)\le  p_B(\mB)\cup\mB$. Indeed, if Spoiler chooses a tuple $\vec a\in R^\mA$, then Duplicator has to answer by either the tuple $f(\vec a)$, or a sequence $\vec b_1,\ldots,\vec b_n\in B^r$ of tuples such that $f(\vec a)=\widehat{p_B}(\vec b_1,\ldots,\vec b_n)$; 
in the first case she loses if $f(\vec a)\not\in R^\mB$, and in the second case she loses if $\vec b_i\not\in R^\mB$ for some $i\in [n]$.
Thus if Duplicator 
has a winning strategy in the one round game and $\mB\in\cK$ for some $\cP$-closed quantifier $Q_\cK$, then $f(\mA)\in\cK$, and since $f$ is an isomorphism  $\mA\to f(\mA)$, also $\mA\in\cK$. In other words, if $\mB\models Q_\cK\vec y\,(R_1(\vec y),\ldots,R_m(\vec y))$, then $\mA\models Q_\cK\vec y\,(R_1(\vec y),\ldots,R_m(\vec y))$. The reverse implication is obtained by using the move described above with the structures switched.

By allowing only $k$ variables and repeating rounds indefinitely (unless Spoiler wins at some round), we obtain a game such that Duplicator having a winning strategy implies $\mA\equiv^k_{\infty\omega,\cP}\mB$. However, in order to prove the converse implication we need to modify the rules explained above. This is because $p_B(\mB)\cup\mB$ is not necessarily closed with respect to the function $p_B$, and in the argument above it would equally well suffice that $f(\mA)\le\mC$ for some structure $\mC$ that is obtained by applying $p_B$ repeatedly to $\mB$. In the next definition we formalize the idea of such repeated applications.

\begin{definition}
Let $p\colon A^n\to A$ be a partial function, and let $R\subseteq A^r$.
We define a sequence $\Gamma^i_p(R)$, $i\in\omega$, of $r$-ary relations on $A$ by the following recursion:
\begin{itemize}
    \item $\Gamma^0_p(R):= R$;  $\Gamma^{i+1}_p(R):=p(R)\cup\Gamma^i_p(R)$.
\end{itemize}
Furthermore, we define $\Gamma^\omega_p(R)=\bigcup_{i\in\omega}\Gamma^i_p(R)$.

This is generalized to $\tau$-structures in the natural way: for all $i\in\omega\cup\{\omega\}$,
$\Gamma^i_p(\mA)$ is the $\tau$-structure $\mC$ such that $C=A$ and  $R^\mC:=\Gamma^i_p(R^\mA)$ for each $R\in\tau$.
\end{definition}

Note that since $\Gamma^i_p(R)\subseteq\Gamma^{i+1}_p(R)$ for all $i\in\omega$ (assuming $A$ is finite) there exists $j\le |A^r|$ such that $\Gamma^\omega_p(R)=\Gamma^j_p(R)$. Similarly for any finite structure $\mA$, $\Gamma^\omega_p(\mA)=\Gamma^j_p(\mA)$ for some $j\le |A^r|$, where $r$ is the maximum arity of relations in $\mA$.

\begin{lemma}\label{P-closed-char}
Let $\cP$ a family of $n$-ary partial functions.
A quantifier is $\cP$-closed if and only if the implication
$$  
    \text{$\mB\in\cK$ and
$\mA\le\Gamma^\omega_{p_A}(\mB)\;\Longrightarrow\;\mA\in\cK$}
$$  
holds for all structures $\mA$ and $\mB$ with $A=B$.
\end{lemma}

\begin{proof}
Assume first that $Q_\cK$ is $\cP$-closed, $\mB\in\cK$ and
$\mA\le\Gamma^\omega_{p_A}(\mB)$. We show first by induction on $i$ that $\Gamma^i_{p_A}(\mB)\in\cK$ for all $i\in\omega$. For $i=0$ this holds by assumption. If $\Gamma^i_{p_A}(\mB)\in\cK$, then $\Gamma^{i+1}_{p_A}(\mB)=p_A(\mC)\cup\mC$, for $\mC=\Gamma^i_{p_A}(\mB)$, and hence $\Gamma^{i+1}_{p_A}$ follows from the assumption that $Q_\cK$ is $\cP$-closed. 

As noted above, there exists $j\in\omega$ such that $\Gamma^\omega_{p_A}(\mB)=\Gamma^j_{p_A}(\mB)$. Thus we have $\mA\le\Gamma^j_{p_A}(\mB)\le\Gamma^{j+1}_{p_A}(\mB)=p_A(\Gamma^j_{p_A}(\mB))\cup\Gamma^j_{p_A}(\mB)$.
Since $\Gamma^j_{p_A}(\mB)\in\cK$ and $\cK$ is $\cP$-closed, it follows that $\mA\in\cK$.

Assume then that the implication  
\begin{quote}
      $(*)$ \quad $\mB\in\cK$ and
$\mA\le\Gamma^\omega_{p_A}(\mB)\;\Longrightarrow\;\mA\in\cK$
\end{quote}
holds for all $\mA$ and $\mB$ with $A=B$. Assume further that $\mB\in\cK$ and $\mA\le p_A(\mB)\cup\mB$. By definition $p_A(\mB)\cup\mB=\Gamma^1_{p_A}(\mB)$, and since $\Gamma^1_{p_A}(\mB)\le\Gamma^\omega_{p_A}(\mB)$, we have $\mA\le\Gamma^\omega_{p_A}(\mB)$. Thus $\mA\in\cK$ follows from the implication $(*)$.
\end{proof}

\subsection{Game for $\cP$-closed quantifiers}

We are now ready to give the formal definition of our pebble game for $\cP$-closed quantifiers.
Let $k$ be a positive integer.
Assume that $\mA$ and $\mB$ are 
$\tau$-structures for a relational vocabulary $\tau$. Furthermore, assume that
$\alpha$ and $\beta$ are assignments on $\mA$ and $\mB$, respectively, such that 
$\dom(\alpha)=\dom(\beta)\subseteq X$, where $X=\{x_1,\ldots,x_k\}$. 
The \emph{$k$-pebble $\cP$ game} for $(\mA,\alpha)$ and 
$(\mB,\beta)$ is played between \emph{Spoiler} and \emph{Duplicator}. We denote the game 
by $\PG_k^\cP(\mA,\mB,\alpha,\beta)$, and we use the shorthand notation 
 $\PG_k^\cP(\alpha,\beta)$ whenever $\mA$ and $\mB$ are clear from the context.

\begin{definition}\label{modified-game} 
The rules of the game $\PG^\cP_k(\mA,\mB,\alpha,\beta)$ 
are the following: 
\begin{enumerate}
\item If $\alpha\mapsto\beta\notin\PI(\mA,\mB)$, then the game ends, and Spoiler wins.

\item If (1) does not hold, there are two types of moves that Spoiler can choose to play: 
\begin{itemize}

\item {\bf Left $\cP$-quantifier move:} Spoiler starts by choosing $r\in [k]$ 
and an $r$-tuple 
$\vec y\in X^r$ of distinct variables. 
Duplicator responds with 
a bijection $f\colon B\to A$. Spoiler answers by choosing an $r$-tuple 
$\vec b\in B^r$. Duplicator answers by choosing $P\subseteq A^r$
such that $f(\vec b)\in\Gamma^\omega_{p_A}(P)$. 
Spoiler completes the round by choosing $\vec a\in P$.
The players continue by playing $\PG^\cP_k(\alpha',\beta')$, 
where $\alpha':=\alpha[\vec a/\vec y]$ and $\beta':=\beta[\vec b/\vec y]$.

\item {\bf Right $\cP$-quantifier move:} Spoiler starts by choosing $r\in [k]$ 
and an $r$-tuple 
$\vec y\in X^r$ of distinct variables. 
Duplicator chooses next 
a bijection $f\colon A\to B$. Spoiler answers by choosing an $r$-tuple 
$\vec a\in A^r$. Duplicator answers by choosing $P\subseteq B^r$
such that $f(\vec a)\in\Gamma^\omega_{p_B}(P)$. 
Spoiler completes the round by choosing $\vec b\in P$.
The players continue by playing $\PG^\cP_k(\alpha',\beta')$,
where $\alpha':=\alpha[\vec a/\vec y]$ and $\beta':=\beta[\vec b/\vec y]$.
\end{itemize}

\item Duplicator wins the game if Spoiler does not win it in a finite number of rounds.
\end{enumerate}
\end{definition}

We now prove that the game $\PG^\cP_k$ indeed characterizes equivalence of structures with respect to the infinitary $k$-variable logic with all $\cP$-closed quantifiers. 

\begin{theorem}\label{fullCSPGchar}
Let $\cP$ be an invariant family of 
partial functions. Then Duplicator has a winning strategy in $\PG^\cP_k(\mA,\mB,\alpha,\beta)$ if, and only if, $(\mA,\alpha)\equiv^{k}_{\infty\omega,\cP}(\mB,\beta)$.
\end{theorem}

\begin{proof}
$\Rightarrow$: We prove by induction on $\varphi\in L^k_{\infty\omega}(\bQ_\cP)$ 
that (for any assignments $\alpha$ and $\beta$) if Duplicator 
has a winning strategy in $\PG^\cP_k(\alpha,\beta)$, then 
$(\mA,\alpha)\models\varphi\iff(\mB,\beta)\models\varphi$.
\begin{itemize}
\item If $\varphi$ is an atomic formula, the claim follows from the fact that Spoiler always 
wins the game $\PG^\cP_k(\alpha,\beta)$ immediately if $\alpha\mapsto\beta\notin\PI(\mA,\mB)$.

\item The cases $\varphi=\lnot\psi$,  $\varphi=\bigvee\Psi$ and $\varphi=\bigwedge\Psi$ 
are straightforward.

\item By Proposition~\ref{exists-P-closed}, the negation of the existential quantifier is in $\bQ_\cP$,
and hence we do not need to consider the case $\varphi=\exists x_i\psi$ separately.

\item Consider then the case $\varphi=Q_\cK \vec y\,\cI$ 
for some $r$-ary 
quantifier $Q_\cK\in\bQ_\cP$ and interpretation
$\cI=(\psi_1,\ldots,\psi_\ell)$. 
We start by assuming that $(\mA,\alpha)\models\varphi$. 
Thus, $\cI(\mA,\alpha):=(A,R_1,\ldots,R_\ell)\in\cK$. 
Let Spoiler play in the game $\PG^\cP_k(\alpha,\beta)$
a left $\cP$-quantifier move with $r$ and the tuple $\vec y\in X^r$, and let  
$f\colon B\to A$ be the bijection given by the winning strategy of Duplicator. Let $\cI(\mB,\beta):=(B,R'_1,\ldots,R'_\ell)$, and for each
$i\in[\ell]$, let $S_i:=f(R'_i)$. We claim that $\mD:=(A,S_1,\ldots,S_\ell)\in\cK$. Since $f$ is an isomorphism 
$\cI(\mB,\beta)\to\mD$, it follows then that $(\mB,\beta)\models\varphi$.  

To prove the claim it suffices to show that $\mD\le\Gamma^\omega_{p_A}(\cI(\mA,\alpha))$,
since then $\mD\in\cK$ by Lemma~\ref{P-closed-char} and the assumption that $Q_\cK$ is $\cP$-closed.
To show this, let $i\in [\ell]$ and $\vec a\in S_i$. We let Spoiler choose the
tuple $\vec b=f^{-1}(\vec a)$ as his answer to the bijection $f$. Thus, $(\mB,\beta[\vec b/\vec y])\models\psi_i$.
Let $P\subseteq A^r$ be the answer of Duplicator. Then by the rules of the game
$\vec a\in\Gamma^\omega_{p_A}(P)$, and Duplicator has 
a winning strategy in the game $\PG^\cP_k(\alpha[\vec a/\vec y],\beta[\vec b/\vec y])$
for all $\vec a\in P$. Hence by induction hypothesis 
$(\mA,\alpha[\vec a/\vec y])\models\psi_i$, i.e., $\vec a\in R_i$,
holds for all $\vec a\in P$. This shows that $S_i\subseteq \Gamma^\omega_{p_A}(R_i)$, and since this holds for all $i\in [\ell]$, we see that $\mD\le\Gamma^\omega_{p_A}(\cI(\mA,\alpha))$.

By using the right $\cP$-quantifier move in place of the left quantifier move, we can prove  
that $(\mB,\beta)\models\varphi$ implies $(\mA,\alpha)\models\varphi$. Thus, 
$(\mA,\alpha)\models\varphi\iff(\mB,\beta)\models\varphi$, as desired.
\end{itemize}

$\Leftarrow$: Assume then that $(\mA,\alpha)\equiv^{k}_{\infty\omega,\cP}(\mB,\beta)$. Clearly it suffices to show that Duplicator can play in the first round of the game $\PG^\cP_k(\alpha,\beta)$ in such a way that $(\mA,\alpha')\equiv^{k}_{\infty\omega,\cP}(\mB,\beta')$ holds, where $\alpha'$ and $\beta'$ are the assignments arising from the choices of Spoiler and Duplicator. 

Assume first that Spoiler decides to play a left $\cP$-quantifier move in the first round of $\PG^\cP_k(\alpha,\beta)$. Let $\vec y\in X^r$ be the tuple of variables he chooses. Since $A$ and $B$ are finite, for each $\vec a\in A^r$ there is a formula $\Psi_{\vec a}\in L^k_{\infty\omega}(\bQ_\cP)$ such that for any $\tau$-structure $\mC$ of size at most $\max\{|A|,|B|\}$, any assignment $\gamma$ on $\mC$, and any tuple $\vec c\in C^r$ we have 
\begin{itemize}
\item $(\mA,\alpha[\vec a/\vec y])\equiv^{k}_{\infty\omega,\cP}(\mC,\gamma[\vec c/\vec y])$ if and only if $(\mC,\gamma[\vec c/\vec y])\models\Psi_{\vec a}$.
\end{itemize}

Let $\vec c_1,\ldots,\vec c_\ell$ be a fixed enumeration of the set $A^r$, and let $\cI$ be the interpretation $(\Psi_1,\ldots,\Psi_m)$, where $\Psi_j:=\Psi_{\vec c_j}$ for each $j\in [m]$.
We define $\cK$ to be the closure of the class $\{\mD\mid \mD\le\Gamma^\omega_{p_A}(\cI(\mA,\alpha))\}$ under isomorphisms. Note that if 
$\mD\le\Gamma^\omega_{p_A}(\cI(\mA,\alpha))$ and $\mE\le\Gamma^\omega_{p_A}(\mD)$, then clearly $\mE\le\Gamma^\omega_{p_A}(\cI(\mA,\alpha))$. Hence by Lemma~\ref{P-closed-char}, the quantifier $Q_\cK$ is $\cP$-closed. Moreover, since $\cI(\mA,\alpha)\in\cK$, we have $(\mA,\alpha)\models Q_\cK\vec y\, \cI$, and consequently by our assumption, $(\mB,\beta)\models Q_\cK\vec y\, \cI$. Thus, there is a structure $\mD\le\Gamma^\omega_{p_A}(\cI(\mA,\alpha))$ and an isomorphism $f\colon \cI(\mB,\beta)\to\mD$. We let Duplicator to use the bijection $f\colon B\to A$ as her answer to the choice $\vec y$ of Spoiler.

Let $\vec b\in B^r$ be the answer of Spoiler to $f$, and let $\vec a=f(\vec b)$. Clearly $(\mA,\alpha)\models \forall\vec y\, \bigvee_{j\in [\ell]}\Psi_j$, whence there exists $j\in [\ell]$ such that $(\mB,\beta[\vec b/\vec y])\models\Psi_j$, or in other words, $\vec b\in R^{\cI(\mB,\beta)}_j$. Since $f$ is an isomorphism $\cI(\mB,\beta)\to\mD$, we have $\vec a\in R_j^\mD$. 
We let Duplicator to use $P:=R^{\cI(\mA,\alpha)}_j$ as her answer to the choice $\vec b$ of Spoiler; this is a legal move since $\mD\le\Gamma^\omega_{p_A}(\cI(\mA,\alpha))$.
Observe now that since $P=R_j^{\cI(\mA,\alpha)}$, we have $(\mA,\alpha[\vec a/\vec y])\models\Psi_{\vec c_j}$, and consequently $(\mA,\alpha[\vec c_j/\vec y])\equiv^{k}_{\infty\omega,\cP}(\mA,\alpha[\vec a/\vec y])$, for all $\vec a\in P$. On the other hand we also have $(\mB,\beta[\vec b/\vec y])\models\Psi_{\vec c_j}$, and hence $(\mA,\alpha[\vec c_j/\vec y])\equiv^{k}_{\infty\omega,\cP}(\mB,\beta[\vec b/\vec y])$. Thus the condition $(\mA,\alpha')\equiv^{k}_{\infty\omega,\cP}(\mB,\beta')$, where $\alpha'=\alpha[\vec a/\vec y]$ and $\beta'=\beta[\vec b/\vec y]$, holds after the first round of $\PG^\cP_k(\alpha,\beta)$ irrespective of the choice $\vec a\in P$ of Spoiler in the end of the round. 
\smallskip

The case where Spoiler starts with a right $\cP$-quantifier move is handled in the same way by switching the roles of $(\mA,\alpha)$ and $(\mB,\beta)$.
\end{proof}

\section{Playing the game}\label{sec:CFI}
In this section we use the game $\PG^k_{\cP}$ to show inexpressibility of a property of finite structures in the infinitary finite variable logic $L^\omega_{\infty\omega}$ augmented by all $\cN_\ell$-closed quantifiers.  More precisely, we prove that the Boolean constraint satisfaction problem $\CSP(\mC_\ell)$, where $\mC_\ell$ is the structure with $C=\{0,1\}$ and two $\ell$-ary relations $R_0 = \{(b_1,\ldots,b_{\ell}) \mid \sum_{i\in [\ell]} b_i \equiv 0 \pmod 2\}$ and $R_1 = \{(b_1,\ldots,b_{\ell}) \mid \sum_{i\in [\ell]} b_i \equiv 1 \pmod 2\}$, is not definable in $L^\omega_{\infty\omega}(\bQ_{\cN_\ell})$. 

\subsection{CFI construction}

In the proof of the undefinability of $\CSP(\mC_\ell)$ we use a variation of the well-known CFI construction, due to Cai, Fürer and Immerman \cite{CFI92}. Our construction is a minor modification of the one that was used in \cite{Hella96} for producing non-isomorphic structures on which Duplicator wins the $k,n$-bijection game. 
We start by explaining the details of the construction.

Let $G=(V,E,\le^G)$ be a connected $\ell$-regular ordered graph. For each vertex $v\in V$, we use the notation
$E(v)$ for the set of edges adjacent to $v$ and $\vec
e(v)=(e_1,\ldots, e_\ell)$ for the tuple that lists $E(v)$ in the
order $\le^G$. The CFI structures we use have in the universe two
elements $(e,1)$ and $(e,2)$ for each $e \in E$, and two $\ell$-ary relations that connect such pairs $(e,i)$ for edges $e$ that are adjacent to some vertex $v\in V$. 

\begin{definition}
Let $G=(V,E,\le^G)$ be a connected $\ell$-regular ordered graph and let $U\subseteq V$. 
We define a CFI structure $\mA_\ell(G,U)=(A_\ell(G),R_0^{\mA_\ell(G,U)},R_1^{\mA_\ell(G,U)})$, 
where $\ar(R_0)= \ar(R_1)=\ell$, as follows.
\begin{itemize}
\item $A_\ell(G):=E\times [2]$,
\item $R_0^{\mA_\ell(G,U)}:=\bigcup_{v\in V\setminus U}R(v)\cup\bigcup_{v\in U}\tilde{R}(v)$ and $R_1^{\mA_\ell(G,U)}:=\bigcup_{v\in  U}R(v)\cup\bigcup_{v\in V\setminus U}\tilde{R}(v)$, where 
\begin{itemize}
\item $R(v):= \{((e_1,i_1),\ldots,(e_\ell,i_\ell))\mid (e_1,\ldots,e_\ell)=\vec e(v),\, \sum_{j\in [\ell]}i_j =0 \pmod 2\}$, and
\item $\tilde{R}(v):= \{((e_1,i_1),\ldots,(e_\ell,i_\ell))\mid (e_1,\ldots,e_\ell)=\vec e(v),\, \sum_{j\in [\ell]}i_j = 1 \pmod 2\}$.
\end{itemize}
\end{itemize}
\end{definition}

For each $v\in V$, we denote the set $E(v)\times [2]$ by $A(v)$. Furthermore, we define
$\mA_\ell(v):=(A(v),R(v),\tilde{R}(v))$ and $\tilde{\mA}_\ell(v):=(A(v),\tilde{R}(v),R(v))$.

By a similar argument as in the CFI structures constructed in \cite{Hella96} and \cite{Hella23} it can be proved that $\mA_\ell(G,U)$ and $\mA_\ell(G,U')$ are isomorphic if and only if $|U|$ and $|U'|$ are of the same parity.
%
%
%
%
%
We choose $\mA_\ell^\even(G):=\mA_\ell(G,\emptyset)$ and 
$\mA_\ell^\odd(G):=\mA_\ell(G,\{v_0\})$
as representatives of these two isomorphism classes, 
where $v_0$ is the least element of $V$ with respect to the linear order $\le^G$. 
We show first that these structures are separated by $\CSP(\mC_\ell)$. 

\begin{lemma}\label{parity-lemma}
$\mA_\ell^\even(G)\in\CSP(\mC_\ell)$, but $\mA_\ell^\odd(G)\not\in\CSP(\mC_\ell)$.
\end{lemma}

\begin{proof}
Let $h\colon A_\ell(G)\to\{0,1\}$ be the function such that $h((e,1))=1$ and $h((e,2))=0$ for all $e\in E$. Then for any tuple $((e_1,i_1),\ldots,(e_\ell,i_\ell))$ the parity of $\sum_{j\in [\ell]}h((e_j,i_j))$ is the same as the parity of $\sum_{j\in [\ell]}i_j$. Thus, $h$ is a homomorphism $\mA_\ell^\even(G)\to\mC_\ell$.

To show that $\mA_\ell^\odd(G)\not\in\CSP(\mC_\ell)$, assume towards contradiction that $g\colon A_\ell(G)\to\{0,1\}$ is a homomorphism $\mA_\ell^\odd(G)\to\mC_\ell$. Then for every $e\in E$ necessarily $g((e,1))\not=g((e,2))$. Furthermore, for every $v\in V\setminus\{v_0\}$, the number $n_v:=|\{e\in E(v)\mid g((e,2))=1\}|$ must be even, while the number $n_{v_0}$ must be odd. Thus, $\sum_{v\in V}n_v$ must be odd. However, this is impossible, since clearly $\sum_{v\in V}n_v=2|\{e\in E\mid g((e,2))=1\}|$.
\end{proof}

\subsection{Good bijections}

Our aim is to prove, for a suitable graph $G$, that Duplicator has a
winning strategy in 
$\PG^{\cN_\ell}_k(\mA_\ell^\even(G),\mA_\ell^\odd(G),\emptyset,\emptyset)$. For
the winning strategy, Duplicator needs a collection of well-behaved
bijections. We define such a collection $\gb$ in
Definition~\ref{good-bij} below.  One requirement is that the bijections preserve the first component of the elements $(e,i)\in A_\ell(G)$.
\begin{definition}
    A bijection $f\colon A_\ell(G)\to A_\ell(G)$ is \emph{edge preserving} if for every $e\in E$ and 
$i\in [2]$, $f((e,i))$ is either $(e,1)$ or $(e,2)$. 
\end{definition}

For any edge preserving $f$ and any $v\in V$
we denote by $f_v$ the restriction of $f$ to the set $E(v)\times [2]$. The \emph{switching number} $\swn(f_v)$ of $f_v$ is $|\{e\in E(v)\mid f_v((e,1))=(e,2)\}|$.
The lemma below follows directly from the definitions of $\mA_\ell(v)$ and $\tilde{\mA}_\ell(v)$.

\begin{lemma}\label{swn-lemma}
Let $f\colon A_\ell(G)\to A_\ell(G)$ be an edge preserving bijection, and let $v\in V$.

(a) If $\swn(f_v)$ is even, then $f_v$ is an automorphism of the structures 
$\mA_\ell(v)$ and $\tilde{\mA}_\ell(v)$.

(b) If $\swn(f_v)$ is odd, then $f_v$ is an isomorphism between the structures
$\mA_\ell(v)$ and $\tilde{\mA}_\ell(v)$.
\end{lemma}

Given an edge preserving bijection $f\colon A_\ell(G)\to A_\ell(G)$ we denote by 
$\Odd(f)$ the set of all $v\in V$ such that $\swn(f_v)$ is odd. Observe that 
$|\Odd(f)|$ is necessarily even.

\begin{corollary}
An edge preserving bijection $f\colon A_\ell(G)\to A_\ell(G)$ is an automorphism of the 
structures $\mA_\ell^\even(G)$ and $\mA_\ell^\odd(G)$ if and only if $\Odd(f)=\emptyset$.
\end{corollary}

\begin{proof}
If $\Odd(f)=\emptyset$, then by Lemma~\ref{swn-lemma}(a) $f_v$ is an automorphism of $\mA_\ell(v)$ and $\tilde{\mA}_\ell(v)$ for all $v\in V$. Clearly this means that $f$ is an automorphism of $\mA_\ell^\even(G)$ and $\mA_\ell^\odd(G)$. On the other hand, if $v\in\Odd(f)$, then by Lemma~\ref{swn-lemma}(b), for any tuple $\vec a\in A(v)^\ell$, we have $\vec a\in R(v)\iff f(\vec a)\in\tilde{R}(v)$. Since $R(v)\cap\tilde{R}(v)=\emptyset$, it follows that $f$ is not an automorphism of $\mA_\ell^\even(G)$ and $\mA_\ell^\odd(G)$.
\end{proof}

\begin{definition}\label{good-bij}
Let $f\colon A_\ell(G)\to A_\ell(G)$ be edge preserving bijection. 
Then $f$ is a \emph{good bijection} if either
$\Odd(f)=\emptyset$ or $\Odd(f)=\{v_0,v\}$ for some $v\in V\setminus\{v_0\}$. 
We denote the set of all good bijections by $\gb$.
\end{definition}

Note that if $f\colon A_\ell(G)\to A_\ell(G)$ is a good bijection, then there is exactly one
vertex $v^*\in V$ such that $f_{v^*}$ is not a partial isomorphism 
$\mA_\ell^\even(G)\to\mA_\ell^\odd(G)$. In case $\Odd(f)=\emptyset$, $v^*=v_0$, while
in case $\Odd(f)=\{v_0,v\}$ for some $v\not=v_0$, $v^*=v$. We denote this vertex
$v^*$ by $\tw(f)$ (the \emph{twist} of $f$).

Assume now that Duplicator has played a good bijection $f$ in the game $\PG^{\cN_\ell}_k$ on the structures $\mA_\ell^\even(G)$ and $\mA_\ell^\odd(G)$. Then it is sure that Spoiler does not win the game in the next position $(\alpha,\beta)$ if $(e,1)$ and $(e,2)$ are not in the range of $\alpha$ (and $\beta$) for any $e\in E(\tw(f))$. This leads us to the following notion.

\begin{definition}\label{gbF}
Let $f$ be a \emph{good} bijection, and let $F\subseteq E$. Then $f$ if \emph{good
for $F$} if $E(\tw(f))\cap F=\emptyset$.
We denote the set of all bijections that are good for $F$ by $\gb(F)$.
\end{definition}

\begin{lemma}\label{gb-lemma}
If $f\in\gb(F)$, then $f\restriction (F\times [2])$ is a partial isomorphism 
$\mA_\ell^\even(G)\to\mA_\ell^\odd(G)$.
\end{lemma}

\begin{proof}
Clearly $f\restriction (F\times [2])\subseteq \bigcup_{v\in V\setminus\{\tw(f)\}}f_v$. By Lemma~\ref{swn-lemma}, $f_v$ is an automorphism of $\mA_\ell(v)$ for any $v\in V\setminus\{\tw(f),v_0\}$, and if $v_0\not=\tw(f)$, $f_{v_0}$ is an isomorphism $\mA_\ell(v)\to\tilde{\mA}_\ell(v)$. The claim follows from this.
\end{proof}

Given a good bijection $f$ with $\tw(f)=u$ and an $E$-path $P= (u_0,\ldots,u_m)$ 
from $u=u_0$ to $u'=u_m$, we obtain a new edge preserving bijection $f_P$ by switching $f$
on the edges $e_i:=\{u_i,u_{i+1}\}$, $i<m$, of~$P$: $f_P((e_i,j))=(e_i,3-j)$ for $i<m$,
and $f_P(c)=f(c)$ for other $c\in A_\ell(G)$. Clearly $f_P$ is also a good bijection, and
$\tw(f_P)=u'$.

\subsection{Cops and Robber game}

In order to prove that Duplicator has a winning strategy in 
$\PG^{\cN_\ell}_k(\mA_\ell^\even(G),\mA_\ell^\odd(G),\emptyset,\emptyset)$ we need to assume that
the graph $G$ has a certain largeness property with respect the number~$k$. We formulate 
this largeness property in terms of a game, $\CR^\ell_k(G)$, that is a new variation of the 
\emph{Cops\&Robber games} used for similar purposes in \cite{Hella96} and \cite{Hella23}.

\begin{definition}\label{CRgame}
The game $\CR^\ell_k(G)$ is played between two players, Cop and Robber. The positions of 
the game are pairs $(F,u)$, where $F\subseteq E$, $|F|\le k$, and $u\in V$.  
The rules of the game are the following:
\begin{itemize}
\item Assume that the position is $(F,u)$.
\item If $E(u)\cap F\not=\emptyset$, the game ends and Cop wins.
\item Otherwise Cop chooses a set $F'\subseteq E$ such that $|F'|\le k$. Then Robber answers
by giving mutually disjoint $E\setminus (F\cap F')$-paths $P_i=(u,u_1^i,\ldots,u_{n_i}^i)$, 
$i\in [\ell]$, from $u$ to vertices $u_i:=u_{n_i}^i$; here mutual disjointness means that 
$P_i$ and $P_{i'}$ do not share edges for $i\not=i'$ (i.e., $u^i_1\not=u^{i'}_1$ and
$\{u_j^i,u_{j+1}^i\}\not=\{u_{j'}^i,u_{j'+1}^i\}$ for all $j$ and $j'$). 
Then Cop completes the round by choosing $i\in [\ell]$. The next position is $(F',u_i)$.
\end{itemize}
\end{definition}

The intuition of the game $\CR^\ell_k(G)$ is that Cop has $k$ pebbles that he plays on 
edges of $G$ forming a set $F\subseteq E$; these pebbles mark the edges $e$ such that $(e,1)$ or $(e,2)$ is in the range of $\alpha$ or $\beta$ in a position $(\alpha,\beta)$ of the game $\PG^{\cN_\ell}_k$ on $\mA_\ell^\even(G)$ and $\mA_\ell^\odd(G)$.
Robber has one pebble that she plays on the vertices of $G$; this pebble marks the vertex $\tw(f)$, where $f$ is the good bijection played by Duplicator in the previous round of $\PG^{\cN_\ell}_k$. 

Cop captures Robber and wins the game if after some round (at least) one of his pebbles is on
an edge that is adjacent to the vertex containing Robber's pebble. This corresponds to a position $(\alpha,\beta)$ in the game $\PG^{\cN_\ell}_k$ such that $\alpha\mapsto\beta$ is potentially not a partial isomorphism. Otherwise Lemma~\ref{gb-lemma} guarantees that $\alpha\mapsto\beta$ is a partial isomorphism. Cop can then move any number of his pebbles to new positions on $G$. While the pebbles
Cop decides to move are still on their way to their new positions, Robber is allowed to prepare $\ell$ mutually disjoint escape routes along edges of $G$ that do not contain any stationary pebbles of Cop. We show in the proof of Theorem~\ref{winning-game} that these escape routes generate tuples $\vec a_1,\ldots,\vec a_\ell$ such that $f(\vec b)=\hat{q}(\vec a_1,\ldots,\vec a_\ell)$, where $q=n^\ell_{A_\ell(G)}$ and $\vec b$ is the tuple chosen by Spoiler after Duplicator played $f$. This gives Duplicator a legal answer $P=\{a_1,\ldots,\vec a_\ell\}$ to $\vec b$. Then Spoiler completes the round by choosing one of the tuples in $P$.
Correspondingly, in the end of the round of $\CR^\ell_k(G)$ Cop chooses which escape route Robber has to use by blocking all but one of them. 

\begin{definition}
Assume that $u\in V$ and $F\subseteq E$ is a set of edges such that $|F|\le k$.
We say that $u$ is \emph{$k$-safe for $F$} if Robber has a winning strategy in
the game $\CR^\ell_k(G)$ starting from position $(F,u)$.
\end{definition}

We prove next the existence of graphs $G$ such that Robber has a winning strategy in the game $\CR^\ell_k(G)$.

\begin{theorem}\label{G-existence}
  For every $\ell \ge 3$ and every $k \geq 1$, there is an $\ell$-regular graph $G = (V,E)$ such that every vertex $v\in V$ is $k$-safe for $\emptyset$.
\end{theorem}

\begin{proof}
Clearly if Robber has a winning strategy in $\CR^\ell_k(G)$, it also has a winning strategy in $\CR^\ell_{k'}(G)$ for $k' < k$.  Thus, it suffices to prove the theorem for  $k \ge \ell$.
  
  By a well-known result of Erd\"os and Sachs~\cite{ErdosS63}, there exist $\ell$-regular connected graphs of \emph{girth} $g$ for arbitrarily large $g$.  Choose a positive integer $d$ with $d > \frac{\log 2k}{\log (\ell -1)} +1$ and let $G$ be an $\ell$-regular graph of girth $g > 6d$.  We claim that any vertex $v$ in $G$ is $k$-safe for $\emptyset$.

  To prove this, we show inductively that Robber can maintain the following invariant in any position $(F,u)$ reached during the game:
\begin{description}  
\item[$(*)$]   for each edge $e \in F$, neither end point of $e$ is within distance $d$ of $u$ in $G$.      
\end{description} 
  Note that, from the assumption that $k  \geq \ell$ and $d > \frac{\log 2k}{\log (\ell -1)}  +1$, it follows that $d \ge 2$.  Thus, the invariant $(*)$ guarantees, in particular, that Cop does not win at any point.

  Clearly the invariant $(*)$ is satisfied at the initial position, since $F$ is empty.  Suppose now that it is satisfied in some position $(F,u)$ and Cop chooses a set $F'$ in the next move.  Let $C \subseteq V$ denote the set of end points of all edges in $F'$.  Since $ |F'| \leq k$, we have $|C| \leq 2k$.

  Let $N \subseteq V$ denote the collection of vertices which are at distance at most $3d$ from $u$.  By the assumption on the girth of $G$, the induced subgraph $G[N]$ is a tree.  We can consider it as a rooted tree, with root $u$.  Then, $u$ has exactly $\ell$ children.  All vertices in $N$ at distance less than $3d$ from $u$ have exactly $\ell -1$ children (and one parent), and all vertices at distance exactly $3d$ from $u$ are leaves of the tree.  This allows us to speak, for instance, of the subtree rooted at a vertex $u'$ meaning the subgraph of $G$ induced by the vertices $x$ in $N$ such that the unique path from $u$ to $x$ in $G[N]$ goes through $u'$.

Let $u_1,\ldots,u_{\ell}$ be the children of $u$.  For each $i$, let $U_i$ denote the set of descendants of $u_i$ that are at distance exactly $d$ from $u$ (and so at distance $d-1$ from $u_i$).  Note that the collection $U_1,\ldots,U_{\ell}$ forms a partition of the set of vertices in $N$ that are at distance exactly $d$ from $u$.  Each $x \in U_i$ is the root of a tree of height $2d$.  Moreover, since the tree below $u_i$ is $(\ell -1)$-regular, $U_i$ contains exactly $(\ell -1)^{d-1}$ vertices.  By the assumption that $d > \frac{\log 2k}{\log (\ell -1)} +1$, it follows that $(\ell -1)^{d-1} > 2k \ge |C|$ and therefore each $U_i$ contains at least one vertex $x_i$ such that the subtree rooted at $x_i$ contains no vertex in $C$.  Let $y_i$ be any descendant of $x_i$ at distance $d$ from $x_i$ and let $P_i$ denote the unique path in $G[N]$ from $u$ to $y_i$.  Robber's move is to play the paths $P_1,\ldots,P_{\ell}$.  We  now verify that this is a valid move, and that it maintains the required invariant $(*)$.

First, note that the paths  $P_1,\ldots,P_{\ell}$ are paths in the tree $G[N]$ all starting at $u$ and the second vertex in path $P_i$ is $u_i$.  It follows that the paths are pairwise edge disjoint.  We next argue that no path $P_i$ goes through an edge in $F \cap F'$. Indeed, by the inductive assumption, no endpoint of an edge in $F$ appears within distance $d$ of $u$ and therefore the path from $u$ to $x_i$ does not go through any such vertex.  Moreover, by the choice of $x_i$, no endpoint of an edge in $F'$ appears in the subtree rooted at $x_i$ and therefore the path from $x_i$ to $y_i$ does not go through any such vertex.  Together these ensure that the path $P_i$ does not visit any vertex that is an endpoint of an edge in $F \cap F'$.

Finally, to see that the invariant $(*)$ is maintained, note that all vertices that are at distance at most $d$ from $y_i$ are in the subtree of $G[N]$ rooted at  $x_i$.  The choice of $x_i$ means this contains no vertex in $C$.  This is exactly the condition that we wished to maintain.
\end{proof}

\subsection{Winning the game}

We are now ready to prove that a winning strategy for Robber in
$\CR^\ell_k(G)$ generates a winning strategy for Duplicator in  the game $\PG^{\cN_\ell}_k$ on the structures $\mA_\ell^\even(G)$ and $\mA_\ell^\odd(G)$.

\begin{theorem}\label{winning-game}
Let $G$ be a connected $\ell$-regular ordered graph.
If $v_0$ is $k$-safe for the empty set, then Duplicator has a winning strategy in the game 
$\PG^{\cN_\ell}_k(\mA_\ell^\even(G),\mA_\ell^\odd(G),\emptyset,\emptyset)$.
\end{theorem}

\begin{proof}
We show that Duplicator can maintain the following invariant for all positions $(\alpha,\beta)$
obtained during the play of the game 
$\PG^{\cN_\ell}_k(\mA_\ell^\even(G),\mA_\ell^\odd(G),\emptyset,\emptyset)$:
\begin{description}
\item[($\dagger$)] There exists a bijection $f\in\gb(F_\alpha)$ such that 
$p:=\alpha\mapsto\beta\subseteq f$ and $\tw(f)$  is $k$-safe for $F_\alpha$, 
where $F_\alpha:=\{e\in E\mid \ran(\alpha)\cap \{e\}\times [2]\not=\emptyset\}$.
\end{description}
Note that if ($\dagger$) holds, then $p\subseteq f\restriction (F_\alpha\times [2])$
and, by Lemma~\ref{gb-lemma}, $f\restriction(F_\alpha\times [2]) \in \PI(\mA_\ell^\even(G),\mA_\ell^\odd(G))$, whence Spoiler does not win the game
in position $(\alpha,\beta)$. Thus, maintaining the invariant ($\dagger$) during the play
guarantees a win for Duplicator.

Note first that ($\dagger$) holds in the initial position $(\alpha,\beta)=(\emptyset,\emptyset)$
of the game:
if $f_0\in\gb$ is the bijection with $\tw(f_0)=v_0$, as
$\emptyset\mapsto\emptyset=\emptyset\subseteq f_0$ and $\tw(f_0)$
is $k$-safe for $F_\emptyset=\emptyset$.

Assume then that ($\dagger$) holds for a position $(\alpha,\beta)$, and assume that Spoiler plays a left 
$\cN_\ell$-quantifier move by choosing $r\le k$ and $\vec y\in X^r$. 
Duplicator answers this by giving the bijection $f^{-1}$. Let 
$\vec b=(b_i,\ldots,b_r)\in A_\ell(G)^r$ be the second part of Spoiler's move, and let 
$F'$ be the set $\{e\in E\mid \ran(\beta[\vec b/\vec y])\cap \{e\}\times [2]\not=\emptyset\}$. 
Since $\tw(f)$ is $k$-safe for $F_\alpha$, there
are mutually disjoint $E\setminus (F_\alpha\cap F')$-paths $P_i$, $i\in [\ell]$,
from $\tw(f)$ to some vertices $u_i$ that are $k$-safe for the set $F'$. Let
$f_{P_i}$, $i\in [\ell]$, be the good bijections obtained from $f$ as explained before
Definition~\ref{gbF}. Now Duplicator answers the move $\vec b$ of Spoiler
by giving the set $P=\{\vec a_1,\ldots,\vec a_\ell\}$ of $r$-tuples, where $\vec a_i:= f^{-1}_{P_i}(\vec b)$ for each $i\in [\ell]$. 

To see that this is a legal move,
observe that since the paths $P_i$ are disjoint, for each $j\in [r]$ there is at most
one $i\in [\ell]$ such that $f^{-1}_{P_i}(b_j)\not= f^{-1}(b_j)$. Thus we have
$\hat{q}(\vec a_1,\ldots,\vec a_\ell)=f^{-1}(\vec b)$, and hence $f^{-1}(\vec b)\in q(P)\subseteq\Gamma^\omega_{q}(P)$ for $q=n^\ell_{A_\ell(G)}$, as required. 
Let Spoiler complete
the round of the game by choosing $i\in [\ell]$; thus, the next position is 
$(\alpha',\beta'):=(\alpha[\vec a_i/\vec y],\beta[\vec b/\vec y])$. It suffices now to show 
that ($\dagger$) holds for the position $(\alpha',\beta')$ and the bijection $f':=f_{P_i}$.

Note first that $F_{\alpha'}=F'$, since
clearly $\ran(\alpha[\vec a_i/\vec y])\cap \{e\}\times [2]\not=\emptyset$ if, and only if,
$\ran(\beta[\vec b/\vec y])\cap \{e\}\times [2]\not=\emptyset$. Thus, $\tw(f')=u_i$
is $k$-safe for $F_{\alpha'}$.  This implies that $f'\in\gb(F_{\alpha'})$, since
otherwise by Definition~\ref{CRgame}, Cop would win the game $\CR_k(G)$ immediately
in position $(F_{\alpha'},\tw(f'))$. It remains to show that $p':=\alpha'\mapsto\beta'$
is contained in $f'$. For all components $a^j_i$ of $\vec a_i$ we have 
$p'(a^j_i)=b_j=f'(a^j_i)$ by definition of $\vec a_i$. On the other hand, for any element 
$a\in\dom(p')\setminus\{a^1_i,\ldots,a^r_i\}$ we have $p'(a)=p(a)=f(a)$.
Furthermore, since the path $P_i$ does not contain any edges in 
$F_\alpha\cap F_{\alpha'}$, we have $f'\restriction (F_\alpha\cap F_{\alpha'})\times [2]
=f\restriction (F_\alpha\cap F_{\alpha'})\times [2]$, and since clearly 
$a\in(F_\alpha\cap F_{\alpha'})\times [2]$, we see that
$f'(a)=f(a)$. Thus, 
$p'(a)=f'(a)$.

The case where Spoiler plays a right $\cN_\ell$-quantifier move is similar.
\end{proof}

Note that the vocabulary of the structures $\mA_\ell^\even(G)$ and $\mA_\ell^\odd(G)$ consists of two $\ell$-ary relation symbols. The presence of at least $\ell$-ary relations is actually necessary: Duplicator cannot have a winning strategy in $\PG^{\cN_\ell}_{\ell-1}$ on structures containing only relations of arity less than $\ell$, since by Corollary~\ref{less-than-ell-ary}(b), all properties of such structures are definable in $L^{\ell-1}_{\infty\omega}(\bQ_{\cN_\ell})$.

From Lemma~\ref{parity-lemma}, Theorem~\ref{G-existence} and Theorem~\ref{winning-game}, we immediately obtain the result.

\begin{theorem}
For any $\ell\ge 3$, $\CSP(\mC_\ell)$ is not definable in $L^\omega_{\infty\omega}(\bQ_{\cN_\ell})$. 
\end{theorem}

Note that $\CSP(\mC_\ell)$ corresponds to solving systems of linear equations over $\Z/2\Z$ with all equations containing (at most) $\ell$ variables. Thus, as a corollary we see that solvability of such systems of equations cannot be expressed in $L^\omega_{\infty\omega}(\bQ_{\cN_\ell})$ for any $\ell$. Furthermore, since systems of linear equations over $\Z/2\Z$ can be solved in polynomial time, we see that the complexity class $\PTIME$ is not contained in $L^\omega_{\infty\omega}(\bQ_{\cN_\ell})$ for any $\ell$.

Finally, note that since the class $\CSP(\mC_\ell)$ is downwards monotone, by Lemma~\ref{r-N-ell-closed} the quantifier $Q_{\CSP(\mC_\ell)}$ is $\cN_{\ell+1}$-closed. Thus, we get the following hierarchy result for the near-unanimity families $\cN_\ell$ with respect to the arity $\ell$ of the partial functions.

\begin{theorem}
For every $\ell\ge 3$ there is a quantifier in $\bQ_{\cN_{\ell+1}}$ which is not definable in $L^\omega_{\infty\omega}(\bQ_{\cN_\ell})$. 
\end{theorem}

\section{Conclusion}\label{Conc}

We have introduced new methods, in the form of pebble games, for
proving inexpressibility in logics extended with generalized
quantifiers.  There is special interest in proving inexpressibility in
logics with quantifiers of unbounded arity.  We introduced a general
method of defining such collections inspired by the equational
theories of polymorphisms arising in the study of constraint
satisfaction problems.  Perhaps surprisingly, while the collection of
CSP that have near-unanimity polymorphisms is rather limited (as they
all have bounded width), the collection of quantifiers with the
corresponding closure property is much richer, including even CSP that are intractable.
The pebble game gives a general method of proving inexpressibility that works for a wide variety of closure conditions.  We were able to deploy it to prove that solvability of systems of equations over $\Z/2\Z$ is not definable using only quantifiers closed under near-unanimity conditions.

It would be interesting to use the pebble games we have defined to show undefinability with other collections of quantifiers closed under partial polymorphisms.  Showing some class is not definable with quantifiers closed under partial Maltsev polymorphisms would be especially instructive.  It would require using the pebble games with a construction that looks radically different from the CFI-like constructions most often used.  This is because CFI constructions encode problems of solvability of equations over finite fields (or more generally finite rings), and all of these problems are Maltsev-closed.

\bibliography{../ref.bib}

\begin{thebibliography}{10}

\bibitem{AtseriasBD09}
A.~Atserias, A.~Bulatov, and A.~Dawar.
\newblock Affine systems of equations and counting infinitary logic.
\newblock {\em Theoretical Computer Science}, 410(18):1666--1683, 2009.

\bibitem{BartoKW17}
L.~Barto, A.~A. Krokhin, and R.~Willard.
\newblock Polymorphisms, and how to use them.
\newblock In A.~A. Krokhin and S.~Zivn{\'{y}}, editors, {\em The Constraint
  Satisfaction Problem: Complexity and Approximability}, volume~7 of {\em
  Dagstuhl Follow-Ups}, pages 1--44. Schloss Dagstuhl - Leibniz-Zentrum
  f{\"{u}}r Informatik, 2017.
\newblock \href {https://doi.org/10.4230/DFU.Vol7.15301.1}
  {\path{doi:10.4230/DFU.Vol7.15301.1}}.

\bibitem{Bulatov17}
A.~A. Bulatov.
\newblock A dichotomy theorem for nonuniform {CSP}s.
\newblock In {\em 58th {IEEE} Annual Symposium on Foundations of Computer
  Science, {FOCS}}, pages 319--330, 2017.
\newblock \href {https://doi.org/10.1109/FOCS.2017.37}
  {\path{doi:10.1109/FOCS.2017.37}}.

\bibitem{CFI92}
J-Y. Cai, M.~F\"{u}rer, and N.~Immerman.
\newblock An optimal lower bound on the number of variables for graph
  identification.
\newblock {\em Combinatorica}, 12(4):389--410, 1992.

\bibitem{DawarGL22}
A.~Dawar, E.~Gr{\"{a}}del, and M.~Lichter.
\newblock Limitations of the invertible-map equivalences.
\newblock {\em arXiv}, abs/2109.07218, 2021.
\newblock URL: \url{https://arxiv.org/abs/2109.07218}.

\bibitem{DawarGP19}
A.~Dawar, E.~Gr{\"{a}}del, and W.~Pakusa.
\newblock Approximations of isomorphism and logics with linear-algebraic
  operators.
\newblock In {\em 46th International Colloquium on Automata, Languages, and
  Programming, {ICALP} 2019.}, pages 112:1--112:14, 2019.
\newblock \href {https://doi.org/10.4230/LIPIcs.ICALP.2019.112}
  {\path{doi:10.4230/LIPIcs.ICALP.2019.112}}.

\bibitem{DawarGHL09}
A.~Dawar, M.~Grohe, B.~Holm, and B.~Laubner.
\newblock Logics with rank operators.
\newblock In {\em Proc. 24th IEEE Symp. on Logic in Computer Science}, pages
  113--122, 2009.

\bibitem{DawarH17}
A.~Dawar and B.~Holm.
\newblock Pebble games with algebraic rules.
\newblock {\em Fundam. Inform.}, 150(3-4):281--316, 2017.

\bibitem{EbbinghausF99}
H-D. Ebbinghaus and J.~Flum.
\newblock {\em Finite Model Theory}.
\newblock Springer, 2nd edition, 1999.

\bibitem{ErdosS63}
P.~Erdos and H.~Sachs.
\newblock Regul{\"a}re graphen gegebener taillenweite mit minimaler knotenzahl.
\newblock {\em Wiss. Z. Martin-Luther-Univ. Halle-Wittenberg Math.-Natur.
  Reihe}, 12(251-257):22, 1963.

\bibitem{FederV98}
T.~Feder and M.Y. Vardi.
\newblock Computational structure of monotone monadic {S}{N}{P} and constraint
  satisfaction: A study through {D}atalog and group theory.
\newblock {\em SIAM Journal of Computing}, 28:57--104, 1998.

\bibitem{GraedelP19}
E.~Gr{\"{a}}del and W.~Pakusa.
\newblock Rank logic is dead, long live rank logic!
\newblock {\em J. Symb. Log.}, 84:54--87, 2019.
\newblock \href {https://doi.org/10.1017/jsl.2018.33}
  {\path{doi:10.1017/jsl.2018.33}}.

\bibitem{GL22}
J.~A. Grochow and M.~Levet.
\newblock On the descriptive complexity of groups without abelian normal
  subgroups.
\newblock {\em arXiv}, abs/2209.13725, 2022.
\newblock \href {https://doi.org/10.48550/arXiv.2209.13725}
  {\path{doi:10.48550/arXiv.2209.13725}}.

\bibitem{Hella96}
L.~Hella.
\newblock Logical hierarchies in {P}{T}{I}{M}{E}.
\newblock {\em Information and Computation}, 129:1--19, 1996.

\bibitem{Hella23}
L.~Hella.
\newblock The expressive power of {CSP}-quantifiers.
\newblock In {\em 31st {EACSL} Annual Conference on Computer Science Logic,
  {CSL}}, pages 25:1--25:19, 2023.
\newblock \href {https://doi.org/10.4230/LIPIcs.CSL.2023.25}
  {\path{doi:10.4230/LIPIcs.CSL.2023.25}}.

\bibitem{HellaI98}
L.~Hella and H.~Imhof.
\newblock Enhancing fixed point logic with cardinality quantifiers.
\newblock {\em J. Log. Comput.}, 8(1):71--86, 1998.
\newblock \href {https://doi.org/10.1093/logcom/8.1.71}
  {\path{doi:10.1093/logcom/8.1.71}}.

\bibitem{Lichter21}
M.~Lichter.
\newblock Separating rank logic from polynomial time.
\newblock In {\em 36th Annual {ACM/IEEE} Symposium on Logic in Computer
  Science, {LICS}}. {IEEE}, 2021.
\newblock \href {https://doi.org/10.1109/LICS52264.2021.9470598}
  {\path{doi:10.1109/LICS52264.2021.9470598}}.

\bibitem{Zhuk20}
D.~Zhuk.
\newblock A proof of the {CSP} dichotomy conjecture.
\newblock {\em J. {ACM}}, 67:30:1--30:78, 2020.
\newblock \href {https://doi.org/10.1145/3402029} {\path{doi:10.1145/3402029}}.

\end{thebibliography}

\end{document}